\documentclass[letter,11pt]{article}

\usepackage{amsfonts,amsmath,amssymb,times}
\usepackage[margin=1in]{geometry}
\usepackage{xspace}
\usepackage{dsfont}
\usepackage[in]{fullpage}
\usepackage[colorlinks=true,citecolor=blue]{hyperref}
\usepackage{enumitem}
\usepackage{multirow}
\usepackage{latexsym}
\usepackage{amsmath,amsfonts,amsthm,amssymb}
\usepackage{verbatim}
\usepackage[dvipsnames,usenames]{color}
\usepackage[textsize=scriptsize]{todonotes}
\usepackage[dvipsnames,usenames]{color}
\usepackage[textsize=scriptsize]{todonotes}
\usepackage{hyperref}
\usepackage{cleveref}

\usepackage[ruled,vlined]{algorithm2e}

\newtheorem{theorem}{Theorem}%[section]
\newtheorem{lemma}[theorem]{Lemma}%[section]
\newtheorem{observation}{Observation}

\newenvironment{proofof}[1]{\begin{proof}[Proof of #1]}{\end{proof}}

\newcommand{\vbp}{\text{\sc vbp}}

\newcommand{\opt}{\text{\sc Opt}}

\newcommand{\A}{\text{\sc A}}
\newcommand{\ccp}{\text{\sc ccp}}

\makeatletter
\newcommand{\E}{\mathbb{E}\xspace}
\newcommand{\pr}{\mathbf{Pr}}

\title{An Asymptotic Lower Bound for Online Vector Bin Packing}
\author{Nikhil Bansal\thanks{CWI Amsterdam and TU Eindhoven, \texttt{N.Bansal@cwi.nl}. Supported by the ERC Consolidator Grant 617951 and the NWO VICI grant 639.023.812.}
\and Ilan Reuven Cohen\thanks{Jether Energy Ltd, Israel, \texttt{ilanrcohen@gmail.com}. This work was done when the author was a postdoc at CWI and supported by the ERC Consolidator Grant 617951.
}
}
\date{}
\begin{document}

\maketitle

\begin{abstract}
We consider the online vector bin packing problem where $n$ items specified by $d$-dimensional vectors
must be packed in the fewest number of identical $d$-dimensional bins.
Azar et al.~\cite{AzarCKS13} showed that for any online algorithm $A$, there exist instances $I$, such that $A(I)$, the number of bins used by $A$ to pack $I$, is $\Omega(d/\log^2 d)$ times $\opt(I)$, the minimal number of bins to pack $I$. However in those instances, $\opt(I)$ was only  $O(\log d)$, which left open the possibility of improved algorithms with better asymptotic competitive ratio when $\opt(I) \gg d$.
We rule this out by showing that for any arbitrary function $q(\cdot)$ and any randomized online algorithm $A$, there exist instances $I$ such that
$  \E[A(I)] \geq c\cdot d/\log^3d \cdot  \opt(I) + q(d)$, for some universal constant $c$.
\end{abstract}

\section{Introduction}
We consider the $d$-dimensional vector bin packing problem (\vbp), a natural and well-studied multi-dimensional generalization of the classic bin packing problem.
Here, we are given a collection $W = \{w_1, w_2,..., w_n\}$ of $n$ items, where each item corresponds to a vector $w_i \in [0,1]^d$. 
The goal is to pack these items feasibly into the fewest number of bins, where a subset 
$S \subseteq W$ of items is said to be packed feasibly in a bin if $\sum_{i\in S} w_i(j) \leq 1$, for each coordinate $j = 1,\ldots,d$.

Note that the case of $d=1$ corresponds to the extensively studied bin-packing problem, where each item (job) corresponds to a number in $[0,1]$
specifying its resource requirement, and the bins correspond to machines. 
In many applications however, jobs require multiple resources and can be modeled as multi-dimensional vectors; with a dimension for each resource. Rescaling the units of measurement so that each machine has unit capacity for each resource, $\vbp$ thus corresponds to finding a feasible assignment of jobs to the least number of machines without exceeding the available resources on any machine.
$\vbp$  has received a lot of attention, both theoretically and practically, especially due to  applications in virtual machine placement in cloud computing, e.g., \cite{panigrahy2011heuristics,zhang2010cloud}. We do not attempt to survey the various results on the problem, and only focus on the ones most relevant to our setting.

In this paper we focus on the online version of \vbp. Here, the vectors arrive one by one, and upon arrival the incoming vector must be assigned immediately and irrevocably, to some (possibly new) bin while maintaining the feasibility of the bin. 
For a randomized online algorithm $A$, given an instance $W$, let $A(W)$ denote the expected number of bins used by $A$, and let $\opt(W)$ denote the optimal number of bins in the optimum offline solution for $W$.
An online algorithm $\A$ is $\alpha$-competitive if for any instance $W$, $A(W) \leq \alpha \cdot \opt(W) + c$, where $c$ is some universal constant, independent of the problems parameters $n,d$.
Under this (strict) notion of competitive ratio, the problem is quite well-understood. Garey et al.~\cite{garey1976resource} showed that the First-Fit algorithm is $(d+0.7)$-competitive ~\cite{coffman1996approximation}.
On the other hand, using a well-known connection of $\vbp$ to graph coloring (that we shall see later), together with a classic result lower bound of Halld\'{o}rsson and Szegedy \cite{halldorsson1992lower} for online coloring, Azar et al.~\cite{AzarCKS13} showed that any randomized algorithm must a have competitive ratio of at least $\Omega(d/\log^2 d)$. They also show more refined and almost tight lower bounds as a function of the ratio of the bin size to the maximum item size.

\vspace{-2mm}

\paragraph{Asymptotic guarantees.} In many bin-packing problems however, it is often more insightful to consider asymptotic guarantees, where the additive error in the approximation (resp.~competitive) ratio for offline (resp.~online) problems can depend on the input parameters.
For example, for the classic bin-packing problem, strictly speaking, it is NP-Hard to find an approximation factor better than $3/2$, as it is hard to distinguish a solution value of $2$ versus $3$, as seen by a simple reduction from the Partition problem. But there exist much better asymptotic guarantees such as $(1+\epsilon) \opt + O(1/\epsilon^2)$, for all $\epsilon >0$, and $\opt + O(\log \opt)$ \cite{VL81, HR17}. These are not only better when $\opt$ is large, but also lead to very interesting algorithmic ideas and insights. 
Asymptotic guarantees for online (1-$d$) bin packing have also been studied extensively. While 
Ullman~\cite{princeton1971performance} proved that for any 
instance $W$, $\text{FirstFit}(W) \leq 1.7\, \opt(W) + 3$, several algorithms are known with improved asymptotic guarantees.  
Currently, the best such result is a $1.5783\cdot(1+\epsilon) \opt(W) + O(1/\epsilon)$ competitive algorithm for any $\epsilon >0$, due to Balog et al.~\cite{balogh2017new}, and the best asymptotic lower bound is $1.5403$~\cite{balogh2012new}.

Similarly, strictly speaking, $\vbp$ is hard to approximate within a factor of $d^{1-\epsilon}$ for any $\epsilon>0$, when $d$ is part of the input. This follows from the connection of $\vbp$ to graph coloring mentioned above and the classic hardness of approximation results for graph coloring \cite{Hastad01, Khot01}. However, and perhaps somewhat surprisingly, much better approximation guarantees of the form $O(\log d)\, \opt + O_d(1)$ are known, when an additive error and running time exponential in $d$ is allowed \cite{chekuri1999multi, BCS09, BEK16}.
Such bounds are quite useful since 
$d$ is a small constant and $\opt$ is large
in typical  practical applications. 

As running time is a not a consideration for online algorithms, these results also suggest the possibility of an {\em asymptotic} online guarantee of the form $O(\log d)\, \opt + O_d(1)$, or even $O(1)\, \opt + O_d(1)$. Moreover, 
the lower bound results of Azar et al.~\cite{AzarCKS13} do not rule out such improved asymptotic bounds either, as the value of $\opt$ in those instances is only a function of $d$ (in fact $\opt$ is only $O(\log d)$).
The dependence of $\opt$ on $d$ also seems inherent in these constructions as they were based on reductions from coloring (we discuss this more later).
For these reasons, it was believed 
that such improved asymptotic bounds might exist,
and there has been interest in finding such an algorithm. This was also our initial motivation for looking at online $\vbp$.

\vspace{-2mm}
\paragraph{Our result.} 
We show that unfortunately, no online algorithms with substantially better bounds than $\tilde{O}(d)$, where $\tilde{O}(\cdot)$ suppresses poly-logarithmic factors, exist
even if asymptotic guarantees are allowed. More formally, we show the following. 
\begin{theorem} 
\label{thm:reltvbp}
Any randomized online algorithm for $\vbp$ must have an asymptotic competitive ratio of at least $\Omega(d/\log^3 d)$.
In particular, given any arbitrary function $q(\cdot)$, there exists instances $I$, such that $ \E[ A(I)] \geq  c(d/\log^3d) \cdot  \opt(I) + q(d)$ for any randomized online algorithm $A$, for some universal constant $c$.
\end{theorem}

The main idea behind Theorem \ref{thm:reltvbp} is an approach to {\em amplify} the instances used in \cite{AzarCKS13}. The lower bound in \cite{AzarCKS13} is based on a reduction from the online coloring problem to $\vbp$, where the number of items $n$ is equal to the number of dimensions $d$.
However, note that any instance where $d$ is some function of $n$, cannot be used to show the asymptotic lower bounds in Theorem \ref{thm:reltvbp}, as already the trivial packing using $n$ bins is already $1$-competitive asymptotically as $n = q(d)$ for some $q$. In particular, to show such a lower bound one needs a family of instances where the value $\opt$ has no dependence $d$. As we discuss in Section \ref{sec:prel} below, such a dependence seems inherent in direct reduction via online coloring.

To get around this problem, instead of reducing $\vbp$ to coloring, we consider a different \emph{copies-coloring problem} ($\ccp$) that lies between the coloring problem and the fractional coloring problem. 
Even though this problem is almost as difficult as online coloring (due to the close relation between chromatic and fractional chromatic numbers \cite{lovasz75}), the flexibility in CCP instances
will allow us to construct $\vbp$ instances where 
$\opt$ can be arbitrarily large for any fixed $d$. 
Our lower bound instances will be similar to those used in \cite{halldorsson1992lower, AzarCKS13}, but modified to work for the copies-coloring problem. We will not work with the notion of fractional chromatic numbers directly as it is unclear how to do this in the online setting, but our analysis will implicitly use the ideas of \cite{lovasz75}. 

\section{Preliminaries}
\label{sec:prel}
We now give the relevant notation and a more formal description of previous ideas.
We first describe a simple offline reduction from graph coloring to $\vbp$, and then discuss the online variant of graph coloring that arises from online $\vbp$.

Recall that in the coloring problem, given an undirected graph $G(V,E)$, the goal is to color each vertex of $G$ using the least number of colors, so that any two adjacent vertices have different color.
Call a $k$-coloring $g: V \rightarrow [k]$ of $G$ feasible if $g(u) \neq g(v)$ for all $(v,u)\in E$. As
the set of vertices assigned a particular color form an independent set, a $k$-coloring is also equivalently viewed as covering of the vertices $V$ by $k$ independent sets. In this case we say that the chromatic number $\chi(G)$ is at most $k$.

\noindent{\bf Reduction.} Given a graph $G$ on $n=|V|$ vertices, consider the following $\vbp$ instance $W$ with $n$ items and $d=n$ dimensions.
For each vertex $v_i$ we create a vector $w_i$ with $w_i(i) = 1$ and $w_i(j) = 1/n$, if $(v_i,v_j)\in E$ and $j<i$, and  $w_i(j)=0$ otherwise.
Here $w_i(k)$ denotes the $k$-th coordinate of $w_i$. 
The following is an easy observation.
\begin{observation}\label{obv:reduction}
For any subset $S\subseteq [n]$, the set of items $W_S =\{w_{i}:i\in S\}$ can be feasibly packed in a bin iff 
  the set of vertices  $V_S=\{ v_{i}:i\in S\}$ is an independent set in $G$.
\end{observation}
\begin{proof}
Suppose, $(v_i,v_j) \in E$ and $i<j$.
As $w_i(i)=1$ and $w_j(i)>0$, the cannot lie in the same bin. Hence the items that can be packed in a bin correspond to independent sets in $G$. Conversely, given any independent set $V_S \subset V$ in $G$, the corresponding items $W_S$ can be packed feasibly in a bin as 
for any coordinate $i \in S$, there is exactly one item $w_i$ with $w_i(i)=1$ and no item $w_j$ with $w_j(i) =1/n$. Moreover, as there are at most $n$ items, for any coordinate $i \notin S$, there can be at most $|S|\leq n$ items with value $w_j(i) =1/n$. So the constraints for every coordinate $i \in [n]$ are satisfied.
\end{proof}

Given the $n^{1-\epsilon}$ hardness of approximation result \cite{Hastad01} for coloring for any $\epsilon >0$, this directly implies a $d^{1-\epsilon}$ hardness for $\vbp$ when $d$ is part of the input. Let us also note that in this reduction the number of dimensions $d$ in the resulting VBP instance, must necessarily be some function of $n$. In particular, even for bounded degree graphs of degree $d'$, graph coloring is known to be $d'/\text{polylog}(d')$ hard to approximate \cite{Chan13}. While if the number of dimensions $d$ in $\vbp$ is constant, then it has a polynomial time algorithm that uses only $O(\log d) \opt + O_d(1)$ bins, giving essentially an $O(\log d)$ approximation.

\noindent{\bf Online coloring and HS instances.}
%\subsubsection*{Online coloring and HS instances.}
Consider the following {\em online} graph coloring problem. The input is a sequence of vertices, and upon the arrival of a vertex all its neighbors to the previously arrived vertices are revealed, and the algorithm must immediately and irrevocably assign this vertex a feasible color.
This problem was studied by Halld\'{o}rsson and Szegedy~\cite{halldorsson1992lower}, who showed that for any $n$ there exists a adversarial strategy that produces distribution over graphs with $n$ vertices and chromatic number $O(\log n)$, but for which any randomized online  algorithm requires at least $\Omega(n/\log n)$ colors. We will call these {\em HS instances}.
Note that the reduction from online coloring described above, produces an online $\vbp$ instance and works for any $d$ (albeit with $n=d$). In particular, this implies  a distribution over instances where $\opt=O(\log d)$ but any online algorithm uses $\Omega(d/\log d)$ bins.

As we will be interested in asymptotic guarantees, we say that for any function $q$, a randomized online algorithm $A$ for $\vbp$  is $(\alpha,q(d))$ competitive, if for any instance $W$ of $\vbp$, it holds that \[\E[A(W)] \leq \alpha \opt(W)  + q(d).\]

\noindent{\bf Fractional coloring.}
We recall the notion of fractional coloring originally introduced by Lov\'{a}sz \cite{lovasz75}. Given a graph $G$, let $\cal{S}$ denote the collection of all its independent sets. Then the fractional chromatic number, denoted $\chi_f(G)$, is the smallest number of colors needed to fractionally cover each vertex by independent sets to an extent of at least $1$. That is, it is the optimum value of the following exponential size linear program.
\[ \min \sum_{S \in \cal{S}} x_S \qquad 
\text{s.t.} \quad \sum_{S \ni v} x_{S} \geq 1 \quad \forall v \in V,  \qquad 
 x_S \in [0,1] \quad  \forall S \in \cal{S}\]
It is well-known that for any graph $G$, the fractional chromatic number is not too far from that chromatic number. In particular, $\chi_f(G) \leq  \chi(G) \leq (1+\ln n)\, \chi_f(G)$ \cite{lovasz75}. 
\section{The asymptotic lower bound}
In this section we prove Theorem \ref{thm:reltvbp}. We begin by defining the copies-coloring problem and list some of its basic properties.
In \S\ref{sec:lb-ccp}  we give the main argument which proves an asymptotic lower bound on the competitive ratio for online copies-coloring, and finally in \S\ref{sec:lb-vbp} we give the result for online $\vbp$ by giving a reduction from copies-coloring.
  
\subsection{The copies coloring problem}
\label{sec:ccp}
In an (offline) instance of the copies coloring problem, denoted by $G^t(V,E)$, the input is an underlying graph $G(V,E)$ and for each vertex $v$, there are $t$ copies $(v,1),\ldots,(v,t)$. 
The goal is to find a coloring of the copies, using the least number of colors, so that (i) each copy of a vertex is assigned a different color, and (ii) any two copies $(u,i)$ and $(v,j)$ are assigned a different color if $(u,v) \in E$. 
In other words, $G^t$ can be viewed as being obtained from $G$, by replacing each vertex $v$ of $G$ by the complete graph $K_t$, and an edge $(u,v)$ in $G$, by the complete bipartite graph $K_{t,t}$.

Given a feasible coloring $g$ of $G^t(V,E)$, for each color $r$, let 
$H_g(r) = \{v \in V: \exists i\in [t] \text{ such that } g(v,i)=r\}$ denote the set of vertices of $V$ for which some (unique) copy is assigned the color $r$. Note that $H_g(r)$ is defined as subset of the vertices of $G$ (and not of vertex-copies in $G^t$). This distinction will be crucial later.
Moreover, $H_g(r)$ is an independent set in $G$.

Let $\chi(G^t)$ denote the chromatic number of $G^t$.  We have the following simple lemma.
\begin{lemma} For any graph $G$, and $G^t$ obtained from $G$ as described above,
	\[\chi_{f}(G) \leq \chi(G^t)/t  \leq \chi(G).\]
\end{lemma}
\begin{proof}
We first show that  $\chi(G^t) \leq t \chi(G)$. This follows directly as any coloring $g$ of $G$ with $k =\chi(G)$ colors can be extended
to a feasible coloring $f$ of $G^t$ using $kt$ colors, by using a fresh set of colors for each of the $t$ copies of $G$, i.e.~consider the coloring $f(v,i) = (g(v),i)$.
	
We now show that $\chi_f(G) \leq \chi(G^t)/t$.
Let $k = |\chi(G^t)|$ and 
let $f$ be some coloring of $G^t$ using $k$ colors.
For $i \in [k]$, consider the independent set $H_f(i)$ of $G$ as defined above.
Consider the fractional coloring $x$ for $G$ with value $k/t$ 
	obtained by setting 
	$x_S = 1/t$ for each of the $k$ sets $S= H_f(i)$ for $i\in [k]$. 
	This fractional coloring is valid as for each $v\in V$, its $t$ copies use different colors, 
	and hence $v$ lies in $t$ independent sets $H_f(i)$.
\end{proof}

The online version of copies coloring problem is defined as follows. There are $n$ time steps, at each step a vertex $v \in G$ arrives,
and reveals all its edges to previously arrived vertices in $G$. At this step,
  $t$ copies $(v,1),\ldots,(v,t)$ of $v$ are created in $G^t$ and all the edges between these copies and the previously created copies are added depending on the neighbors of $v$ in $G$. 
Each one of these $t$ copies of $v$ must be immediately and irrevocably assigned a feasible color.

\subsection{Asymptotic lower bound for online copies coloring}
\label{sec:lb-ccp}
Let $G$ be a graph on $n$ vertices, and $G^t$ be the corresponding graph with copies as defined above.
Let $n = |V(G)|$. We say that an online algorithm  for $G^t$ is $(\alpha,q(n))$-competitive, for some function $q$, if 
the number of colors used by it most $\alpha \cdot \chi(G^t) + q(n)$.

Note that even though $G^t$ has $n \cdot t$ vertices, the additive term  $q(n)$ in the definition above only depends on $n$.
This distinction between $t$ and $n$ will be crucial in our reductions and allow us to translate the non-asymptotic lower bounds for coloring $G$ to asymptotic lower bounds for coloring $G^t$. In particular, we show the following.
%  the form $(\alpha,q(n)$ for coloring $G^t$.

\begin{theorem}
\label{thm:reltcolor}
For any function $q$, 
there exists a family of instances $G^t$, such that there is there is no randomized 
$(o(n/\log^3 n), q(n))$-competitive algorithm for online copies coloring, where $n$ is the number of vertices in the underlying graph $G$.   
\end{theorem}

Consider the instance $G$ obtained from the HS graphs on $n$ vertices 
\cite{halldorsson1992lower}. As stated in Section \ref{sec:prel}, these are $O(\log n)$ colorable, and there is an online realization of the graph, in for which any online algorithm uses $\Omega(n/\log n)$ colors in expectation, and in particular implies an $\Omega(n/\log^2n)$ lower bound for online coloring.

So to prove Theorem \ref{thm:reltcolor}, it suffices to show the following.
\begin{lemma}
    \label{lem:maincolor}
For any function $q$,
if there exists an $(\alpha, q(n))$-competitive algorithm for the online copies coloring problem on instances $G^t$ for some $t=\Omega(q(n) \log n)$, then there exists a $3 \log n \cdot  \alpha$-competitive for the online coloring problem on the underlying instances $G$.  
\end{lemma}
In particular, note that our reduction loses a factor of $O(\log n)$ in $\alpha$.

We now focus on proving Lemma \ref{lem:maincolor}.
Fix some function $q$, and
let $A$  be  some $(\alpha,q(n))$-competitive algorithm for the online copies coloring problem.
% and $q$ the corresponding additive function. 
We will use $A$ to obtain another randomized algorithm $B$ for an online coloring of $G$.
In particular, algorithm $B$ on instance $G = G(V,E)$ simulates $A$ on $G^t$ with $t = q(n) \log n$ copies of each vertex, and 
{\em tries} to color $V$ using some (random) subset of the colors used by~$A$.
	 
We describe the algorithm formally in the figure below, and use the following notation.
% for a online copies coloring algorithm $A$ where each vertex has $T$.
Let $v_1,\ldots,v_n$ denote the online sequence of the vertices of $G$. At step $i \in [n]$,  
let $A(v_i,j)$ be the color  used by $A$ for the $j$-th copy of $v_i$,
and let $R^i(A)$ denote the set of all the colors used by $A$ so far until the end of step $i$. 
Let $B(v_i)$ denote the color used by $B$ for $v_i$. We also have special set of $n$ colors $\{\overline{1},\ldots,\overline{n}\}$, that are not used by $A$, and are reserved for use by $B$, in the occasional event that it fails.  
 %for $r \in R(A)$, $H_A(r) = \{v \in V:\exists i \text { s.t.} A(v,i)=r \}$,	 $H_B(r) = \{v \in V:B(v)=r \}$.

Upon the arrival of a vertex $v_i$, $B$ gives $t$ copies of $v_i$ to algorithm $A$. Then, for each new color $r$ used by $A$,
$B$ adds $r$ to its {\em pool} of colors with probability $p = 2 \cdot \log(n)/t$. If $A$ colors some copy of $v_i$ using some color $r$, where $r$ is in the pool of $B$, then $B$ colors $v_i$ with $r$.  Otherwise, $B$ {\em fails} at step $i$ and uses a brand new color for $v_i$.
	 
\begin{algorithm}[H]
\AlgoDontDisplayBlockMarkers\SetAlgoNoEnd\SetAlgoNoLine
\DontPrintSemicolon % Some LaTeX compilers require you to use \dontprintsemicolon instead
\KwIn{An online graph $G(V,E)$, an algorithm $A$ for copies coloring, and parameters $p,t$.}
\KwOut{A feasible coloring of $G(V,E)$}
Init $B$\;
\For{step $i \gets 1$ \textbf{to} $n$} 
{
color $t$ copies of $v_i$ using $A$\;
\For{$r \in R^i(A)\setminus R^{i-1}(A)$} 
    {
        with probability $p$ add $r$ to  $R^i(B)$
    }
 $P_i \leftarrow \{r\in R^i(B): \exists j\in[t], A(v_i,j)=r\}$\; 
  \If {$P_i \neq \emptyset$}
  {
    set $B(v_i) = r'$, for some $r' \in P_i$\;
  }
  \Else
  {
    declare {\em fail} and set $B(v_i) =\overline{i}$
  }
}
\caption{Online coloring algorithm}
\label{algo:max}
\end{algorithm}

The following simple observations show the correctness of $B$ and the bound the number of colors used.
\begin{observation}
$B$ produces a feasible coloring of $G$.
\end{observation}
\begin{proof}
For a color $r$, let $H_A(r)$ denote the set of vertices in $G$ for which some copy is colored $r$ by $A$, and note that $H_A(r)$ is an independent set in $G$. Similarly, let $H_B(r)$ be the vertices in $G$ colored $r$ by $B$.
For $r \in R(B)$, if $r\in R(A)$ then we have that $H_B(r)\subseteq H_A(r)$ therefore $H_B(r)$ is an independent set. On the other hand, if $r \notin R(A)$, then  $r$ must be one of the special colors and hence $|H_B(r)|=1$ and is an independent set.
\end{proof}

Let us say that $B$ {\em fails}, if it uses a special color at any time step. As $B$ can use at most $n$ colors, and since $R(B)$ is formed by choosing each vertex in $R(A)$ independently with probability $p$, we have that:
\begin{observation}
The expected number of colors used by $B$ is at most $\E[|R(A)|]\cdot p + n \cdot \pr[\text{B 'fails'}]$.
\end{observation}
\begin{observation}
	The probability that $B$ fails is at most $1/n$.
\end{observation}	
\begin{proof}
    Let $I(v_i)$ be the set of colors used by $A$ to color the copies of $v_i$. As $|I(v_i)|= t$,
	\[ \pr[\exists j \in P_i ] = 1-\pr[\forall  r \in I(v_i) : r \notin R^i(B)] = 1-\Big(1- \frac{2 \log n}{t}\Big)^t \geq 1 - 1/n^2.\]
By a union bound for all $v \in V$, the probability that $B$ fails is at most $1/n$.
\end{proof}

Together these imply Lemma \ref{lem:maincolor} as follows.  
\begin{proofof}{Lemma~\ref{lem:maincolor}}	
By the observations above, $B$ produces a valid coloring of $G$, and the expected number of colors used is at most 
$\E(|R(A)|) p + 1$. Using that $p = 2 (\log n)/t$ and that $\E[R(A)] \leq  \alpha \cdot \chi(G^t) + q(n)$ by our assumption on $A$, the expected number of colors used by $B$ is at most 
\[ (\alpha\cdot  \chi(G^t) + q(n)) \cdot (2 \log n)/t + 1 \leq   (2 \log n) (\alpha \cdot t \cdot  \chi(G) + q(n))/t + 1 \leq 2 \log n \cdot \alpha \cdot \chi(G) + 3, \]
where the first inequality uses that $\chi(G^t) \leq t \cdot \chi(G)$ and the second inequality uses that $t \geq q(n) \log n$.  
\end{proofof}

\subsection{The reduction to $\vbp$}
\label{sec:lb-vbp}
To complete the proof of Theorem~\ref{thm:reltvbp}, we will show that if there exists a  $(\alpha,q(d))$-competitive algorithm for online $\vbp$, then there exists an
$(\alpha,q(n))$-competitive algorithm for online $\ccp$.

We give a reduction from $\ccp$ to $\vbp$ similar to the one in Section \ref{sec:prel}, but instead of using $d=nt$ dimensions (the number of vertices in $G^t$), we use $d=n$ dimensions.

Consider a $\ccp$ instance $G^t$ and let $G=G(V,E)$ denote the underlying graph on $n$ vertices.  
For each vertex $v_i \in V$, consider the vector 
$w_i$ with coordinates $w_i(i)=1$ and $w_i(j) = 1/n$ if
$(v_i,v_j)\in E$ and $j<i$, and $0$ otherwise.
For each copy $(v_i,k)$ of vertex $v_i \in V$ and $k\in [t]$, we create a copy $(w_i,k)$ of $w_i$.
The goal in the $\vbp$ instance is to pack the vectors-copies $(w_i,k)$ in the fewest bins. 

The following shows that the resulting $\vbp$ instance has a packing using $r$ bins iff $G^t$ is $r$-colorable. 
\begin{observation}
In any feasible solution to the $\vbp$ instance,
the copies of vectors packed in the same bin, correspond to independent sets in $G^t$. In particular,
(i) each bin contains at most one copy of any $w_i$, and (ii) for any subset $S\subseteq [n]$, one copy of $\{ w_{i}:i\in S\}$ can be packed in a bin iff 
    $\{ v_{i}:i\in S\}$ is an independent set in $G$.
\end{observation}
\begin{proof}
The first property holds as $w_i(i)=1$. 
The second property now follows by Observation \ref{obv:reduction} as for a single bin, copies of a vector do not affect the argument as there can be at most one copy of any vector.
\end{proof}
This observation, together with Theorem \ref{thm:reltcolor}, and the fact that $d=n$ in the reduction above, directly implies Theorem \ref{thm:reltvbp}.

\bibliographystyle{plain}
\bibliography{ref}

\end{document}